\documentclass[showpacs,twocolumn,reprint,aps,pdftexify]{revtex4-1}

\usepackage[latin1]{inputenc}
\usepackage{graphicx}% Include figure files
\usepackage{dcolumn}% Align table columns on decimal point
\usepackage{bm}% bold math
\usepackage[colorlinks,linkcolor=blue,urlcolor=blue,citecolor=blue]{hyperref}% add hypertext capabilities
\input{Qcircuit}
\usepackage{amsthm,amsmath,amssymb,amsfonts}
\usepackage{bbm}
\usepackage{hyperref}
\newtheorem{definition}{Definition}
\theoremstyle{definition}
\newtheorem{theorem}{Theorem}
\theoremstyle{plain}
\newtheorem{lemma}{Lemma}
\DeclareMathOperator{\Tr}{Tr}

\begin{document}

\preprint{APS/123-QED}

\title{Unconditional Security with Decoherence-Free Subspaces}

\author{Elloá B. Guedes}
\author{Francisco M. de Assis}%
\affiliation{
IQuanta -- Institute for Studies in Quantum Computation and Information\\
Federal University of Campina Grande\\
Campina Grande -- Paraíba -- Brazil
}%

\date{\today}% It is always \today, today,

\begin{abstract}
We show how to use decoherence-free subspaces over collective-noise quantum channels to convey classical information in perfect secrecy. We argue that codes defined over decoherence-free subspaces are codes for quantum wiretap channels in which the gain of information by a non-authorized receiver is zero. We also show that if some symmetry conditions are guaranteed, the maximum rate on which these secret communications take place is equal to the Holevo-Schumacher-Westmoreland capacity of the quantum channel.
\end{abstract}

% Decoherence, Quantum Communication and Inf Theory

\maketitle

\section{Introduction}

Preventing errors in quantum information is one of the main objectives of Quantum
Information Theory. Errors can arise from the coupling of the system of interest to the environment, and the subsequent \emph{decoherence} induced by this coupling. Because of the fragile nature of quantum states, decoherence is considered as the main obstacle in the transmission of coherent information \cite{Schlosshauer:DecoherenceBook}.

In the context of Quantum Communications, decoherence is responsible for the information leakage out to the environment in a noisy quantum channel. If secret messages are conveyed through it, at least part of them can be gathered by a non-authorized receiver called \emph{wiretapper}. This situation is undesired in a cryptographic scenario and must be avoided.

Cai et. al~\cite{Cai:QuantumWiretap} and Devetak \cite{Devetak:QuantumSecrecyCapacity} modeled this scenario in the so called \emph{quantum wiretap channels}. They also established the conditions to perform classical information exchange without being deceived by a wiretapper. In their formulation, only codes that minimize the error decoding probability and that maximize the equivocation of the wiretapper are adequate. But the maximum rate in which the secret communication takes place, the \emph{quantum secrecy capacity}, is usually below the ordinary capacity of the channel to convey classical information.

For overcoming decoherence, some good methods have been proposed
such as quantum error-correcting codes, dynamical decoupling, decoherence-free subspaces (DFS), and so on \cite{Byrd:ErrorPreventionLeakageElimination}. Regarding DFS, in particular, if the error operators that affect the qubits have some symmetries, then the qubits will suffer from the same noise in quantum channel and that will compensate the effects, keeping the invariability of these states, i.e., no decoherence takes place in such subspaces \cite{Lidar:DecoherenceFreeSubspaces}. Therefore, these subspaces seem to be suitable to construct codes that minimize the information leakage out to the environment.

So far, some works in the literature already explore the potential of DFS in Quantum Communications. All of them consist of protocols against certain types of collective noise (such as rotation, dephasing, amplitude damping, among others) and consider the use of small DFS (with two or three qubits, for instance) \cite{Gu:DeterministicSecureQuantumCommunications,Majgier:ProtectedSubspacesQuantumInformation,Byrd:UniversalLeakegeElimination,Qin:AmplitudeDampingCollectiveNoise,Dong:DSQC}. Even experimental realization were already implemented aiming at quantum information processing \cite{Viola:ExperimentalRealization,Beige:ImplementationDFS,Kielpinski:ExperimentalRelatization,Kwiat:ExperimentalRealization}. In the perspective of these works, the protection of information means avoiding the loss of coherence, maintaining the fidelity of the quantum states.

In this paper, we investigate more general consequences of using DFS in Quantum Communications from the perspective of \emph{secure message exchange}. We state a formal definition of quantum channels which satisfy the DFS symmetry criteria, define codes over their subspaces, and also establish the conditions to perform secret communications in this scenario. We conclude that any code defined with such characteristics is also a code suitable for quantum wiretap channels. Moreover, the secrecy capacity in such subspaces is equal to the Holevo-Schumacher-Westmoreland (HSW) capacity in these quantum channels. This is a particular case in which the secrecy capacity is maximal.

The rest of this paper is structured as follows. The conditions for quantum privacy, established by Schumacher and Westmoreland \cite{Schumacher:QuantumPrivacyCoherence}, as well as the concepts regarding quantum wiretap channels are recalled in Section \ref{sec:quantumPrivacyWiretap}. The fundamentals on DFS are introduced in Section \ref{sec:dfs}. Our contributions are stated in Section \ref{sec:results}. A detailed example illustrating the results obtained is shown in Section \ref{sec:example}. Lastly, final remarks are presented in Section \ref{sec:remarks}.

\section{Quantum Privacy and Quantum Wiretap Channels}\label{sec:quantumPrivacyWiretap}

Suppose that a sender (Alice) prepares a quantum system $B$ in an initial state $\rho$. Her objective is to send this system to a specific receiver (Bob) through a noisy quantum channel, denoted by the superoperator $\mathcal{E}^B$. This way, the state of the system received by Bob is denoted by $\rho_{\textrm{Bob}} = \mathcal{E}^B(\rho)$.

Due to the presence of the noise, to provide a unitary description of the evolution of $\rho_{\textrm{Bob}}$ along the channel it is necessary to consider the interaction with the \emph{environment} which is initially defined to be in a pure state $\ket{0_E}$. In this case, the superoperator is given by

\begin{eqnarray}
\mathcal{E}^B(\rho) &=& \Tr_E U^{BE} \left(\rho \otimes \outter{0_{\textrm{E}}}{0_{\textrm{E}}} \right) U^{BE\dagger}
\end{eqnarray} where $U^{\textrm{BE}}$ stands for a unitary interaction operation.

Suppose that Alice is using this quantum channel to send classical information to Bob. Alice then prepares the quantum system $B$ into one of the possible states $\rho_{k}$ with a priori probabilities $p_k$. To decode the message received, Bob performs measurements using some \emph{decoding observable}. The amount of classical information conveyed from Alice to Bob, which we will denote $\mathbf{H}_{\textrm{Bob}}$, is governed by the Holevo quantity $\chi^{\textrm{Bob}}$, defined as

\begin{equation}
\chi^{\textrm{Bob}} = S(\rho_{\textrm{Bob}}) - \sum_k p_k S\left(\rho_{\textrm{Bob},k}\right) \label{eq:bobHolevo}
\end{equation} Some considerations about the Holevo quantity in this scenario must be mentioned: ($i$) $\mathbf{H}_{\textrm{Bob}} \leq \chi^{\textrm{Bob}}$ regardless of the decoding observable chosen; and ($ii$) $\mathbf{H}_{\textrm{Bob}}$ can be made arbitrarily close to $\chi^{\textrm{Bob}}$ through a suitable choice of a code and decoding observable.  In this case,  $\chi^{\textrm{Bob}}$ represents an upper bound on the classical information conveyed from Alice to Bob.

When considering the cryptographic purposes of the channel, then the eavesdropper (Eve) must have access to some or all of the environment $E$ with which $B$ interacts. The evolution superoperator $\mathcal{E}^B$ describes all effects of the eavesdropper on the channel or, in other words, all of the efforts to wiretap the channel between Alice and Bob are contained in the interaction operator $U^{BE}$.  This way, the information available to Eve, denoted by $\mathbf{H}_{\textrm{Eve}}$, will be limited by:

\begin{equation}
\chi^{\textrm{Eve}} = S(\rho_{\textrm{Eve}}) - \sum_k p_k S(\rho_{\textrm{Eve},k}) \label{eq:evaHolevo}
\end{equation} The inequality $\mathbf{H}_{\textrm{Eve}} \leq \chi^{\textrm{Eve}}$ holds whether Eve has access or not to the entire environment.

The \emph{quantum privacy} is defined as
\begin{equation}
P = \mathbf{H}_{\textrm{Bob}} - \mathbf{H}_{\textrm{Eve}} \label{eq:privacidadeQtca}
\end{equation} Alice and Bob wish to make $P$ as large as possible. But, they must assume that the eavesdropper is acquiring her greatest possible information from the channel. The \emph{guaranteed privacy}, $P_G = \inf P$, is the infimum over all possible strategies adopted by Eve. Since $\mathbf{H}_{\textrm{Eve}} \leq \chi^{\textrm{Eve}}$, then $P_G \geq \mathbf{H}_{\textrm{Bob}} - \chi^{\textrm{Eve}}$. On the other hand, Alice and Bob will want to use the channel to make the guaranteed privacy $P_G$ as great as possible. Let $\mathcal{P} = \sup P_G$. The best scheme they use aims to make $\mathbf{H}_{\textrm{Bob}}$ close to $\chi^{\textrm{Bob}}$. This way, we denote $\mathcal{P}$ as

\begin{equation}
\mathcal{P} = \chi^{\textrm{Bob}} - \chi^{\textrm{Eve}} \label{eq:quantumPrivacy}
\end{equation}

Despite the characterization of the quantum privacy, it is necessary to characterize schemes that describe how Alice and Bob should proceed to establish the channel properties to perform secure communications without being deceived by a wiretapper. To do so, Cai et al. \cite{Cai:QuantumWiretap} and Devetak \cite{Devetak:QuantumSecrecyCapacity} simultaneously characterized the \emph{quantum wiretap channels}, defined as follows.

\begin{definition} \label{def:quantumWiretapChannel}
A quantum memoryless wiretap channel is described by a pair of superoperators $\mathcal{E}^{B}$ and $\mathcal{E}^{E}$ from a complex Hilbert space $\mathcal{H}$. When Alice sends a quantum state $\omega$ from $\mathcal{H}^{\otimes n}$, Bob receives $\mathcal{E}^{\otimes n, B}(\omega)$ and Eve receives $\mathcal{E}^{\otimes n, E}(\omega)$, where $n$ is the dimension of the input Hilbert space.
\end{definition}

The codes used by the legitimate participants of the communication are characterized in Definition \ref{def:codesQuantumWiretap}.

\begin{definition} \label{def:codesQuantumWiretap}
A set of codewords of length $n$ ($n = \textrm{dim} (\mathcal{H})$) for a set \ $\mathcal{U}$ of classical messages is a set of input states labeled by messages in $\mathcal{U}$, $\Omega(\mathcal{U}) = \left\{ \omega(u) : u \in \mathcal{U} \right\}$, and a decoding measurement of length $n$ for the set $\mathcal{U}$ of messages is a measurement on a length-$n$ output space of the channel with outcome in $\mathcal{U}$, i.e., a set of positive operators $\mathcal{D}_u$, $u \in \mathcal{U}$ with $\sum_{u \in \mathcal{U}} \mathcal{D}_u \leq \mathbbm{1}$. The pair $( \Omega(\mathcal{U}), \left\{ \mathcal{D}_u : u \in \mathcal{U} \right\})$ is called a code of length $n$ for the set $\mathcal{U}$ of messages. The rate of this code is $\frac{1}{n}\log \left| \mathcal{U} \right|$.
\end{definition}

The Figure \ref{fig:quantumWiretap} synthesizes the quantum wiretap channel general idea. Alice creates a quantum state $\omega(u)$ when she wants to send the message $u \in \mathcal{U}$ to Bob. Due to the noise, Bob receives $\mathcal{E}^{\otimes n, B}( \omega(u)) = \Tr_E \left[ \mathcal{E}^{\otimes n} (\rho \otimes \outter{0_E}{0_E}) \right]$ and decodes it using a positive operator-value measurement (POVM) $\left\{ \mathcal{D}_u : u \in \mathcal{U} \right\}$, that results in an estimate $u'$ for $u$. Eve, in turn, receives the state $\mathcal{E}^{\otimes n, E}(\omega(u)) = \Tr_B \left[ \mathcal{E}^{\otimes n} (\rho \otimes \outter{0_E}{0_E}) \right]$ and will try to obtain as much information as possible from the message sent by Alice. To do so, she will try to build a POVM based on the typicality of the states measured by her following a strategy presented in \cite[Sec. 4]{Cai:QuantumWiretap}.

\begin{figure}[h!]
  \centering
  \includegraphics[width=0.45\textwidth]{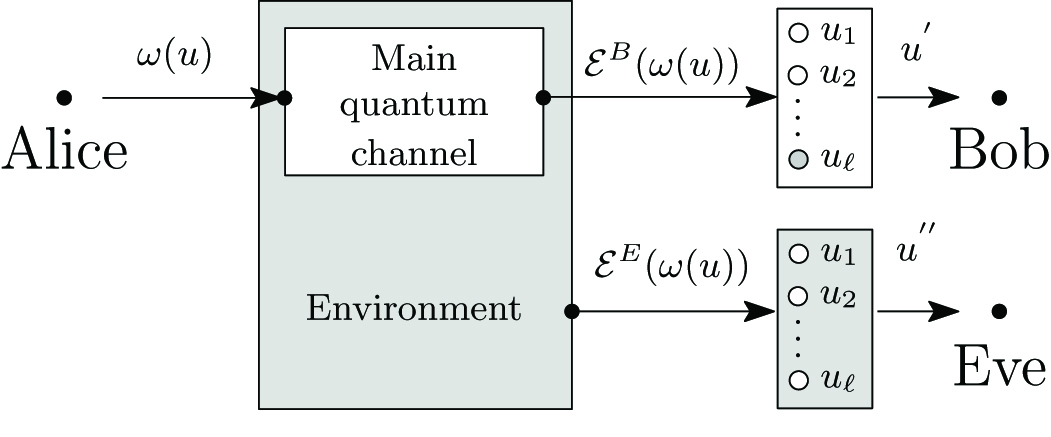}\\
  \caption{General idea of the quantum wiretap channel.}\label{fig:quantumWiretap}
\end{figure}

However, despite the strategy of communication has been depicted, the arguments of security regarding the message exchanged have not been defined yet. It is necessary to ensure a small decoding error probability between the legitimate participants of the communication while the wiretapper learns almost nothing of the secret message. The formalization of these two requirements is presented in Definition \ref{def:wiretapCodes}.

\begin{definition} \label{def:wiretapCodes} (Wiretap Code)
A code $( \Omega(\mathcal{U}), \left\{ \mathcal{D}_u : u \in \mathcal{U} \right\})$ of length $n$ is called a wiretap code with parameters $(n,\left| \mathcal{U} \right|, \lambda, \mu)$ if, for $\lambda, \mu >0$:

\begin{equation}
P_e = 1 - \frac{1}{\left|\mathcal{U}\right|} \sum_{u \in \mathcal{U}} \Tr_E \left[ \mathcal{E}^{\otimes n,B}(\omega(u)) \mathcal{D}_u \right] \leq \lambda \label{eq:quantumWiretapError}
\end{equation} and

\begin{equation}
\footnotesize{\frac{1}{n}\left[ S \left( \sum_{u \in \mathcal{U}} \frac{1}{\left|\mathcal{U}\right|} {\mathcal{E}}^{\otimes n, E}(\omega(u)) \right) - \sum_{u \in \mathcal{U}}\frac{1}{\left|\mathcal{U}\right|} S\left({\mathcal{E}}^{\otimes n ,E}(\omega(u))  \right)\right] < \mu } \label{eq:quantumWiretapSegredo}
\end{equation} where $\frac{1}{n}\log \left| \mathcal{U} \right|$ is called the code rate.
\end{definition}

In this definition, Eq. (\ref{eq:quantumWiretapError}) guarantees that the decoding error probability for Bob is, on average, smaller than $\lambda$; and Eq. (\ref{eq:quantumWiretapSegredo}) bounds the average accessible information in such a way that the wiretapper can obtain (almost) nothing about the message sent by Alice.

Lastly, the \emph{quantum secrecy capacity} is defined as:

\begin{definition}  (\emph{Quantum Secrecy Capacity \cite{Cai:QuantumWiretap}})  The secrecy capacity of a quantum channel is the maximum real number $C_S$ such that for all $\epsilon, \lambda, \mu > 0$ and sufficiently large $n$ there exists a $(n,\left| \mathcal{U} \right|, \lambda, \mu)$ code with:

\begin{equation}
C_S < \frac{1}{n} \log \left| \mathcal{U} \right| + \epsilon
\end{equation}
\end{definition}

Despite the previous definitions assume messages uniformly distributed, the following theorem about the secrecy capacity due to \cite[Sec. 5]{Cai:QuantumWiretap} is a more general result.

\begin{theorem} For a quantum wiretap channel $\mathcal{E}$ as characterized in the Definition \ref{def:quantumWiretapChannel}, the quantum secrecy capacity satisfies:

\begin{equation}
C_{S}(\mathcal{E}) \geq \max_{\left\{ P\right\}} \left[ \chi^{\textrm{Bob}} - \chi^{\textrm{Eve}} \right] \label{eq:quantumSecrecyCapacity}
\end{equation} where the maximum is taken over all probability distributions $P$ over $\mathcal{U}$; and $\chi^{\textrm{Bob}}$ and $\chi^{\textrm{Eve}}$ are the Holevo quantities given in Eqs. (\ref{eq:bobHolevo}) and (\ref{eq:evaHolevo}), respectively.
\end{theorem}

The quantum secrecy capacity can be understood as the maximum rate in which is possible to convey classical information in perfect secrecy through a quantum channel. It is equivalent to the supreme of the guaranteed quantum privacy defined in Eq. (\ref{eq:quantumPrivacy}). A particular characteristic of the quantum secrecy capacity is that it does not have a single-letter characterization, i.e.,  it is not computable since it considers all the input states and also all probability distributions over them \cite{Cai:QuantumWiretap,Devetak:QuantumSecrecyCapacity}.

\section{Decoherence-Free Subspaces} \label{sec:dfs}

Due to decoherence, a quantum system begins to lose energy into the environment and decay to a ground state, its relative phase is erased and, thus, the information it carries is lost \cite{Shabani:DFSRelaxedCriteria}. In this section, we will show how to avoid these undesired effects despite the existence of decoherence.

Let a closed quantum system be composed by the \emph{system of interest} $S$ defined on a Hilbert space $\mathcal{H}$ and by the \emph{environment} $E$. The Hamiltonian that describes this system is defined as follows:

\begin{equation}
\mathbbm{H} = \mathbbm{H}_S \otimes \mathbbm{1}_E + \mathbbm{1}_S \otimes \mathbbm{H}_E + \mathbbm{H}_{SE} \label{eq:hamiltonian}
\end{equation} where $\mathbbm{1}$ is the identity operator; and $\mathbbm{H}_S$, $\mathbbm{H}_E$ and $\mathbbm{H}_{SE}$ denote the Hamiltonians of system, environment and system-environment interaction, respectively.

In order to prevent errors, it would be ideal that  $\mathbbm{H}_{SE}$ were equal to zero, indicating that system and environment are decoupled and evolve independently and unitarily under their respective Hamiltonians $\mathbbm{H}_S$ and $\mathbbm{H}_E$ \cite{Lidar:DecoherenceFreeSubspaces}. However, in practical scenarios, such an ideal situation is not possible since no system is noiseless. So, after isolating a system to the best of our ability, we should aim for the realistic goals of the identification and correction of errors when they occur and/or avoiding noises when possible and/or suppressing noise in the system \cite{Byrd:ErrorPreventionLeakageElimination}.

If some symmetries exist in the interaction between the system and the environment, it is possible to find a ``quiet corner'' in the system Hilbert space not experiencing decoherence. Let $\left\{ A_i (t) \right\}$ be a set of operators in the operator-sum representation (OSR) corresponding to the evolution of the system. We say that a system density matrix $\rho_S$ is \emph{invariant} under the OSR operators $\left\{ A_i(t) \right\}$ if $\sum_i A_i(t) \rho_S A_{i}^{\dagger}(t) = \rho_S$. We are now able to define the DFS whose states are invariant despite a non-trivial coupling between the system and the environment.

\begin{definition}
A subspace $\tilde{\mathcal{H}}$ of a Hilbert space $\mathcal{H}$ is called a DFS with respect to a system-environment coupling if every pure state from this subspace is invariant under the corresponding OSR evolution for any possible environment initial condition:

\begin{equation}
\sum_i A_i(t) |\tilde{k}\rangle\langle\tilde{k}| A_{i}^{\dagger}(t) =  |\tilde{k}\rangle\langle\tilde{k}|, \forall |\tilde{k}\rangle\langle\tilde{k}| \in \tilde{\mathcal{H}}, \forall \rho_E(0)
\end{equation}
\end{definition}

%Besides we have presented a definition of DFS in terms of invariant pure states, a mixed state which has support only over pure states of a DFS will be also be invariant and hence protected from decoherence \cite{Bacon:Tese}.
%
%Quantum systems defined on DFS are totally decoupled from the environment and, for that reason, completely immune to the effects of decoherence. Quantum codes constructed from states of a DFS are classified as \emph{quantum error-avoiding codes} (QEAC) and the tasks of perturbation and recovery on them are trivial \cite{Duan:ErrorAvoidingCodes}.

Let the Hamiltonian of the system-environment interaction be $\mathbbm{H}_{SE} = \sum_j \mathbf{S}_j \otimes \mathbf{E}_j$, where $\mathbf{S}_j$ and $\mathbf{E}_j$ are the system and environment operators, respectively. We consider that the environment operators $\mathbf{E}_j$ are linearly independent. The symmetries required to define a DFS are described in the theorem below. For a detailed proof or different formulations see \cite[Sec.~5]{Lidar:DecoherenceFreeSubspaces}.

\begin{theorem} \label{teo:condicoesDFS} (\emph{DFS Conditions})
A subspace $\tilde{\mathcal{H}}$ is a DFS iff the system operators $\mathbf{S}_j$ act proportional to the identity on the subspace:

\begin{equation}
\mathbf{S}_j | \tilde{k}\rangle = c_j | \tilde{k}\rangle \hspace{0.5cm} \forall j, | \tilde{k}\rangle \in \tilde{\mathcal{H}}
\end{equation}
\end{theorem}

% Eliminei por questão de espaço, para reduzir o texto -- elloa
%In a brief manner, using the notion of quantum error avoiding codes, the conditions for a DFS can also be defined as follows. Let $\left\{ \mathbf{A}_i \right\}$ be a set of Kraus operators representing the errors and let $\rho$ be a pure state. A DFS $\tilde{\mathcal{H}}$ is defined as being the maximal linear subset of $\mathcal{H}$ such that the fidelity $F(\rho, \left\{ \mathbf{A}_i \right\}) = \sum_a \left|\bra{\psi}\mathbf{A}_i \ket{\psi} \right|^2$ is equal to one \cite{Duan:ErrorAvoidingCodes}.

In practice,  identifying a useful symmetry and taking advantage of it can be very difficult. One must ($i$) identify the symmetry, ($ii$) find the states which are invariant to the interaction, and ($iii$) construct, if possible, operations on the system which will serve as a universal set of gating operations while preserving the necessary symmetries \cite{Byrd:ErrorPreventionLeakageElimination}.

Quantum codes constructed from states of a DFS are classified as \emph{quantum error-avoiding codes} (QEAC) and the tasks of perturbation and recovery on them are trivial. They can be contrasted with \emph{quantum-error correcting codes} (QECC) in some aspects. While QECCs are devised to correct errors after their occurrence, QEACs do not have the ability to correct errors since they avoid them; QECCs devised in practical circumstances belong to the class of non-degenerate codes while QEACs are highly degenerate; QEACs usually require a lower number of physical qubits to encode one logical qubit than QECCs. In particular, if the degeneracy attains the maximum, a QECC reduces to a QEAC what illustrates a circumstance in which a type of code becomes equivalent to the other \cite{Duan:ErrorAvoidingCodes}.

The absence of decoherence in DFS has been shown as of major importance for implementations of quantum memory and quantum algorithms. Other applications of it cover encoding information in quantum dots, collective dissipation, noise reduction, among others \cite{Lidar:DecoherenceFreeSubspaces,Byrd:ErrorPreventionLeakageElimination}.

\section{DFS in Secure Quantum Communications} \label{sec:results}

We will now examine the use of DFS in Quantum Communications. We will consider the use of \emph{collective noise quantum channels}, i.e., a model of quantum channels in which several qubits couple identically to the same environment, while undergoing both dephasing and dissipation \cite{Zanardi:CollectiveDecoherence}. This particular case gives a light into some interesting consequences of the use of DFS in quantum communications. Our focus, in particular, will be in the aspects of \emph{secure message exchange}.

We consider the case that Alice wants to convey secret classical messages through a quantum channel to Bob. These messages must be protected from a wiretapper Eve that has full access to the environment. The channel between Alice and Bob has a decoherence-free subspace whose states will be used to encode the secret message. The following definition characterizes such quantum channel.

\begin{definition} \label{def:canalDFS} (Collective Noise Quantum Wiretap Channel)
A collective noise quantum wiretap channel $\mathcal{E}$ is a channel as in Definition \ref{def:quantumWiretapChannel} such that its Kraus decomposition $\left\{ A_i \right\}$ satisfies Theorem \ref{teo:condicoesDFS}.
\end{definition}

Since $\left\{ A_i \right\}$ satisfies Theorem \ref{teo:condicoesDFS}, then the channel $\mathcal{E}$ has a decoherence-free subspace $\tilde{\mathcal{H}}$. When  Alice wants to send a message to a Bob, now she can make it using a QEAC with the following definition.

\begin{definition} \label{def:DFSCode}
Let $\tilde{\mathcal{H}}$ be a DFS spanned by a set of eigenvectors $\{ | \tilde{k}\rangle \}$, i.e., \mbox{$\tilde{\mathcal{H}} = \textrm{Span}[\{ |\tilde{k} \rangle \} ]$}. A set of codewords of length $n$ ($n = \textrm{dim}(\tilde{\mathcal{H}})$) for a set $\mathcal{U}$ of classical messages is a set of input states labeled by messages in $\mathcal{U}$, $\tilde{K}(\mathcal{U}) = \{ \tilde{k}(u) : u \in \mathcal{U} \} \subseteq \tilde{\mathcal{H}}$, and a trivial decoding measurement composed of a set of positive operators $\tilde{\mathcal{D}}_u$, $u \in \mathcal{U}$ with $\sum_{u \in \mathcal{U}} \tilde{\mathcal{D}}_u \leq \mathbbm{1}$. The pair \mbox{$( \tilde{K}(\mathcal{U}), \{ \tilde{\mathcal{D}}_u : u \in \mathcal{U}\})$} is called a QEAC of length $n$ for the set $\mathcal{U}$ of messages. The rate of this code is $\frac{1}{n} \log \left| \mathcal{U} \right|$.
\end{definition}

Without loss of generality, we will consider here that the environment starts in a pure state $\outter{0_\textrm{E}}{0_\textrm{E}}$. This is a clear assumption since we can always imagine that a ``local'' environment in a mixed state is just part of a larger system in a pure entangled state \cite{Schumacher:QuantumPrivacyCoherence}.

Using the code defined, if Alice wants to send a quantum message $u$ now she encodes it the QEAC defined over $\tilde{\mathcal{H}}$, obtaining $\tilde{k}(u)$. When she sends it through the communication channel, the message interacts with the environment. Bob then receives $\rho_{\textrm{Bob}} (\tilde{k}(u))$ and Eve receives $\rho_{\textrm{Eve}}(\tilde{k}(u))$ which are given by:

\begin{eqnarray}
\rho_{\textrm{Bob}} (\tilde{k}(u)) &=& \Tr_\textrm{E} \left[ \mathcal{E}^{\otimes n}(\tilde{k}(u) \otimes \outter{0_E}{0_E}) \right]\\
\rho_{\textrm{Eve}}(\tilde{k}(u)) &=& \Tr_\textrm{B} \left[ \mathcal{E}^{\otimes n}(\tilde{k}(u) \otimes \outter{0_E}{0_E}) \right] \label{eq:finalEva}
\end{eqnarray}

Since Alice used a QEAC, then the existing dynamical symmetry protected the quantum information from the interaction with the environment. It means that the joint evolution of the system and the environment occurred in a decoupled way. Hence, the state $\rho_{\textrm{Bob}} (\tilde{k}(u)) $ is given by:

\begin{align}
\rho_{\textrm{Bob}} (\tilde{k}(u))&= \Tr_\textrm{E} \left[ \mathcal{E}^{\otimes n}(\tilde{k}(u) \otimes \outter{0_E}{0_E}) \right] \label{eq:tracoBob1}\\
&= \Tr_\textrm{E} \left[ \sum_i A_i \left( \tilde{k}(u) \otimes \outter{0_E}{0_E}\right) A_{i}^{\dagger} \right]\\
&= \Tr_\textrm{E} \left[ \tilde{k}(u) \otimes \rho_{E} \right]\label{eq:DFSunitarityConsequence2}\\
&= \tilde{k}(u) \label{eq:bobFinal}
\end{align} where Eq. (\ref{eq:DFSunitarityConsequence2}) is due to the invariance of a state of the DFS under the OSR operators. Taking into account the Hamiltonian of the quantum system given in Eq. (\ref{eq:hamiltonian}) and the fact that the system of interest and the environment did not interact, it is the case that the environment would only suffer from the action of $\mathbbm{H}_E$, indicating a unitary evolution restricted to the environment.

We will now show how such QEAC protects the information conveyed through the channel from a wiretapper.

\begin{lemma} \label{lemma:DFSWiretapCode}
A QEAC as in Definition \ref{def:DFSCode} over a collective noise quantum wiretap channel as in Definition \ref{def:canalDFS} is a wiretap code with parameters $(n, \left| \mathcal{U} \right|, \lambda, \mu)$.
\end{lemma}
\begin{proof}
The proof is straightforward. We have to prove that the QEAC satisfies the criteria of Eqs. (\ref{eq:quantumWiretapError}) and (\ref{eq:quantumWiretapSegredo}).

Let's first analyze the average decoding error probability. Since $\tilde{k}(u)$ is in $\tilde{\mathcal{H}}$, it did not interact with the environment. So, $\rho_{\textrm{Bob}} = \tilde{k}(u)$ as shown in Eqs. (\ref{eq:tracoBob1})-(\ref{eq:bobFinal}). It turns out that the decoding is trivial and that the message sent by Alice can be perfectly recovered since there is a decoding measurement $\tilde{\mathcal{D}}_u$ for every $u \in \mathcal{U}$. We can see, thus, that there is a negligible average decoding error probability for Bob.

Then we proceed to analyze Eq. (\ref{eq:quantumWiretapSegredo}). Recall that it is the average accessible information by Eve which is bounded by the Holevo quantity, defined in Eq. (\ref{eq:evaHolevo}). We will try to obtain the Holevo quantity first.

Despite the state of the environment $\rho_E$ (vide Eq. (\ref{eq:DFSunitarityConsequence2})) is not known, the fact that Alice and Bob used states from a DFS guaranteed that the interaction Hamiltonian $\mathbbm{H}_{SE}$ did not govern the joint evolution of system and environment. Instead of that, each system evolved completely unitary under its own Hamiltonian, i.e., the environment suffered only the action of $\mathbbm{H}_E$. It implies that the environment ended in a pure state. Using this result to calculate the Holevo quantity, we have:

\begin{eqnarray}
\chi^{\textrm{Eve}} &=& S(\rho_{\textrm{Eve}}(\tilde{k}(u))) - \sum_k p_k S(\rho_{\textrm{Eve},k} \tilde{k}(u)) \label{eq:HolevoEvaInicio}\\
&=& S(\rho_E) - \sum_k p_k S(\rho_{\textrm{Eve},k} \tilde{k}(u))\\
&=& 0 - \sum_k p_k S(\rho_{\textrm{Eve},k} \tilde{k}(u))
\end{eqnarray} It is well known that the Holevo quantity $\chi^{\textrm{Eve}} \geq 0$. Since $S(\rho) \geq 0$ for any $\rho$, and that the probabilities $p_k \geq 0$ for any $k$, then it is the case that the remaining term is zero. Thus, $\chi^{\textrm{Eve}} = 0$. Since the Holevo quantity is an upper bound of the accessible information, then Eq. (\ref{eq:quantumWiretapSegredo}) is also equal to zero. It concludes the proof.
\end{proof}

%Another measurement of information that emphasizes the absence of interaction between system of interest and environment is the entropy exchange which is determined entirely by the initial state of the system of interest $B$ and the channel dynamics \cite{Schumacher:QuantumPrivacyCoherence}. In this case, such measure is equal to $S_e = S(\rho_{\textrm{Eve}}(\tilde{k}(u))) = S(\rho_{E}) = 0$ because $\rho_{E}$ is a pure state.  We can conclude, thus, that system and environment are completely decoupled.

We can now characterize the secrecy capacity of a collective noise quantum wiretap channel.

\begin{theorem}
The secrecy capacity of a collective noise quantum wiretap channel $\mathcal{E}$, characterized as in Definition \ref{def:canalDFS}, satisfies

\begin{equation}
C_{S,\textrm{DFS}}(\mathcal{E}) = \max_{\left\{P\right\}} \left[ \chi^{\textrm{Bob}} \right] \label{eq:SecrecyCapacityDFS}
\end{equation} where the maximum is taken over all probability distributions $P$ over $\mathcal{U}$; and $\chi^{\textrm{Bob}}$ is the Holevo quantity given in Eq. (\ref{eq:bobHolevo}).
\end{theorem}
\begin{proof}
Let a QEAC \mbox{$( \tilde{K}(\mathcal{U}), \{ \tilde{\mathcal{D}}_u : u \in \mathcal{U}\})$} be used over the channel $\mathcal{E}$. As proved in Lemma \ref{lemma:DFSWiretapCode}, it was shown that $\chi^{\textrm{Eve}} = 0$. Firstly, this fact is substituted into Eq. (\ref{eq:quantumSecrecyCapacity}). The equality is due to the Holevo-Schumacher-Westmoreland theorem \cite{Schumacher:QuantumChannel}.
\end{proof}

We can therefore conclude that it is possible to perform unconditionally secure quantum communications through wiretapped collective noise quantum channels. The unconditional security criterion is satisfied since $\chi^{\textrm{Eve}} = 0$, meaning that no information was gathered by Eve and that the communication was carried out in perfect secrecy.

The resulting expression of the secrecy capacity of a DFS has a relation to the results presented by Schumacher and Westmoreland \cite{Schumacher:QuantumPrivacyCoherence}. These authors show that the ability of the quantum channel to send private information is at least as great as its ability to send coherent information. Since the information encoded in a DFS does not loose coherence, then its ability to send private information is maximal.

\section{Example -- Collective Dephasing} \label{sec:example}

To illustrate the results described in the previous section, we will show a detailed example of conveying secret classical messages through a collective dephasing quantum channel $\mathcal{E}$. Let's suppose that qubits on this channel couple to the environment in a symmetric manner and undergo a dephasing process, defined as:

\begin{eqnarray}
\ket{0} \rightarrow \ket{0} \hspace{1cm} \ket{1} \rightarrow e^{\imath \phi}\ket{1}
\end{eqnarray}

To overcome decoherence, Alice and Bob can take advantage from an existing symmetry in the channel. If they codify the messages using states that are immune to decoherence, Eve can learn nothing from their secret messages. To do so, Alice and Bob will use the following codification scheme

\begin{eqnarray}
\ket{0_L} = \ket{01} \hspace{1cm}
\ket{1_L} = \ket{10}
\end{eqnarray} A qubit can, thus, be codified as $\ket{\psi_L} = \alpha \ket{0_L} + \beta \ket{1_L}$. It is interesting to see that $\ket{\psi_L}$ does not suffer from the effects of decoherence

\begin{eqnarray}
\mathcal{E}(\ket{\psi_L}) &=& \mathcal{E}\left(\alpha \ket{0_L} + \beta \ket{1_L}\right)\\
&=& \alpha e^{\imath \phi}\ket{01} + \beta  e^{\imath \phi} \ket{10}\\
&=& e^{\imath \phi} \left( \alpha  \ket{01} + \beta \ket{10} \right)\\
&=& e^{\imath \phi} \ket{\psi_L}\\
&=& \ket{\psi_L}
\end{eqnarray} because the overall phase factor $e^{\imath \phi}$ acquired due to the dephasing process has no physical significance.

In this example, the messages sent by Alice are binary, so $\mathcal{U} = \left\{ 0, 1\right\}$. She will encode the bits in the following way: $\tilde{k}(0) = \ket{01}$ and $\tilde{k}(1) = \ket{10}$. So, $\tilde{K}(\mathcal{U}) = \left\{ \ket{01}, \ket{10} \right\}$.  Making use of this code to send the message $u$, Alice encodes  in the corresponding $\tilde{k}(u)$ which she conveys through the channel. We assume that the bits $0$ and $1$ are equally likely and that the environment starts in a pure state.

To decode the states received Bob will use the following POVM: $\tilde{\mathcal{D}}_0 = \outter{01}{01}$ and $\tilde{\mathcal{D}}_1 = \outter{10}{10}$. His Holevo quantity in this scenario will be equal to

\begin{small}
\begin{align}
\chi^{\textrm{Bob}} &= S\left(\rho_{\textrm{Bob}}\tilde{k}(u)\right) - \sum_u p_u S\left(\rho_{\textrm{Bob},u}\right)\\
&=  S\left(\frac{1}{2}\outter{01}{01} + \frac{1}{2}\outter{10}{10}\right)
- \sum_{u = \left\{ 0,1 \right\}} \frac{1}{2} S\left(\tilde{k}(u)\right)\\
&= 1 - 0 - 0 = 1
\end{align}
\end{small} We can conclude that the secrecy capacity for this scenario is equal to $C_{S,DFS}(\mathcal{E}) = 1$ bit per channel use. This is an example of how to convey secret messages using a DFS through a noisy quantum channel with a positive rate using a very simple encoding-decoding scheme.

\section{Final Remarks} \label{sec:remarks}

From our analysis, we can conclude that if a quantum channel is characterized as in Definition \ref{def:canalDFS}, then the existence of certain symmetries in this channel can be exploited to send classical information with unconditional security. The encoding using a QEAC can be seen as an instance of a wiretap code with the particularity that no information is gathered by the wiretapper.

The secrecy capacity of such channels shown in Eq. (\ref{eq:SecrecyCapacityDFS}) is equal to the HSW capacity of a quantum channel \cite{Schumacher:QuantumChannel}. This is a particular case in which the ability of a quantum channel to send secret information can be made as large as its capacity to send ordinary classical information.

Regarding the secrecy capacity, Cai et. al \cite{Cai:QuantumWiretap} argue that obtaining a computable secrecy capacity is most likely to be even more difficult than obtaining a computable form of the capacity of a quantum channel for the classical information. In our particular case for collective-noise quantum wiretap channels, secrecy capacity equals the HSW capacity, what turns out to be less difficult to compute.

Due to the technical difficulties to build completely closed quantum system \cite{Byrd:UniversalLeakegeElimination}, the results shown here can be applied to build devices that perform unconditional secure message exchange even in the presence of decoherence. It is very promising in practical applications especially considering already existing results regarding the use of DFS in communications \cite{Dorner:CommunicationDFS,Jaeger:CommunicationDFS,Xia:CommunicationDFS}, particularly in long-distance \cite{Xue:CommunicationDFS}.

We emphasize that the results presented here cannot be generalized to all quantum channels since the conditions for a DFS are not satisfied by all of them. Zanardi and Rasetti \cite{Zanardi:CollectiveDecoherence} argue that the DFS conditions only arise in scenarios where collective decoherence takes place. Besides that, the advantages verified in terms of secrecy and rate to this particular case are significant.

In future work, we suggest the investigation of the impacts of the results shown in the simplification of certain communication protocols that make use of DFS \cite{Gu:DeterministicSecureQuantumCommunications,Qin:AmplitudeDampingCollectiveNoise,Dong:DSQC}. We also suggest the investigation of more general conditions to the existence of perfect secrecy in quantum systems.

\bibliographystyle{apsrev4-1}
\bibliography{ref}

\end{document}